\documentclass[12pt]{article}
\usepackage[british]{babel}
\usepackage{amsmath,amsfonts,amssymb,amsthm}
\usepackage{graphicx,float}

\usepackage[a4paper=true]{hyperref}

\voffset = - 1.2cm

\newtheorem{theorem}{Theorem}
\newtheorem{definition}[theorem]{Definition}

\newtheorem{corollary}[theorem]{Corollary}
\newtheorem{proposition}[theorem]{Proposition}
\newtheorem{lemma}[theorem]{Lemma}
\newtheorem{example}[theorem]{Example}

\numberwithin{theorem}{section} 
\numberwithin{equation}{section} 
\begin{document}
\baselineskip = 15.5pt
\centerline{\Large\bf  Geometry of Riccati equations over}
\centerline{\Large\bf  normed division algebras}
\medskip
\centerline{J. de Lucas$^*$, M. Tobolski$^*$ and S. Vilari\~no$^{**}$}
\medskip
\medskip

\centerline{$^*$Department of Mathematical Methods in Physics, University of Warsaw,}
\centerline{ul. Pasteura 5, 02-093, Warszawa, Poland}
\medskip
\centerline{$^{**}$Centro Universitario de la Defensa Zaragoza. IUMA. Universidad de
Zaragoza.}
\centerline{
Academia General Militar. C. de Huesca s/n, E-50090, Zaragoza, Spain}
\begin{abstract}

This work presents and studies Riccati equations over finite-dimensional normed division algebras. We prove that a Riccati equation over a finite-dimensional normed division algebra $A$ is a particular case of conformal Riccati equation on an Euclidean space and it can be considered as a curve in a Lie algebra  of vector fields $V\simeq\mathfrak{so}(\dim A+1,1)$.
Previous results on known types of Riccati equations are recovered from a new viewpoint. A new type of Riccati equations, the octonionic Riccati equations, are extended to the octonionic projective line $\mathbb{O}{\rm P}^1$. As a~new physical application, quaternionic Riccati equations are applied to study quaternionic Schr\"odinger equations on 1+1 dimensions.
\end{abstract}
{\bf Keywords}: Lie system, normed division algebra, octonion, quaternion quantum mechanics, Riccati equation, Vessiot--Guldberg Lie algebra.
\\
{\bf MSC}: 34A26, 53Z05, 17B66.
\section{Introduction}
The {\it Riccati equations} \cite{In44,Ri24}, namely the differential equations\footnote{It is sometimes assumed $a(t)c(t)\neq0$. This condition does not change our main results and it only makes some proofs  unnecessarily more technical.}
\begin{equation}\label{RiccR}
\frac{dx}{dt}=a(t)+b(t)x+c(t)x^2, \qquad x\in \mathbb{R},
\end{equation}
with $a(t)$, $b(t), c(t)$ being arbitrary $t$-dependent real functions, play a very relevant r\^ole in physics and mathematics \cite{BLW91,NR02,SJH85}.
Despite its apparent simplicity, the general solution of a Riccati equation for arbitrary $a(t)$, $b(t)$ and $c(t)$ cannot be obtained by quadratures \cite{In44}. This has led to devise many methods to solve them (see e.g. \cite{MH13,Pi12}).

The general solution, $x(t)$, of a Riccati equation (\ref{RiccR}) can be brought into the form
\begin{equation}\label{sup}
x(t)=\frac{x_1(t)(x_3(t)-x_2(t))+kx_2(t)(x_3(t)-x_1(t))}{x_3(t)-x_2(t)+k(x_3(t)-x_1(t))},
\end{equation}
where $x_1(t),x_2(t),x_3(t)$ are different particular solutions to (\ref{RiccR}) and  $k\in\mathbb{R}$ (see \cite{Dissertationes,CGM07,In44,Wi83}). This allows us to
analyse the general solution of a Riccati equation in terms of just three particular ones.
For instance, it follows from (\ref{sup}) that all particular solutions to a Riccati equation are periodic if and only
if the Riccati equation has three periodic particular solutions. Additionally, (\ref{sup}) enables us to simplify  the numerical analysis of Riccati equations \cite{Wi83}.

Expression (\ref{sup}) shows that Riccati equations are Lie systems \cite{LS}. Indeed, recall that a {\it Lie system} is a nonautonomous ($t$-dependent) system of first-order ordinary differential equations whose general solution can be written as a $t$-independent function, the {\it superposition rule}, of a set of particular solutions and some
constants \cite{Dissertationes,LS,Wi83}. In view of (\ref{sup}), the general solution of  (\ref{RiccR}) reads $x(t)=\Phi(x_{(1)}(t),x_{(2)}(t),x_{(3)}(t),k)$, where
$$
\Phi:\mathbb{R}^3\times\mathbb{R}\ni (x_{(1)},x_{(2)},x_{(3)},k)\mapsto x\in \mathbb{R},
$$ and 
$$ x:=\frac{x_{(1)}(x_{(3)}-x_{(2)})+kx_{(2)}(x_{(3)}-x_{(1)})}{x_{(3)}-x_{(2)}+k(x_{(3)}-x_{(1)})}.
$$
The above expression is the most well-known example of nonlinear superposition rule \cite{In44}. This turns Riccati equations into Lie systems.

The Lie--Scheffers Theorem \cite{CGM07,LS} establishes when a system of first-order ordinary differential equations in normal form is a Lie system. It also states that every Lie system amounts to a curve within a finite-dimensional Lie algebra of vector fields: a {\it
Vessiot--Guldberg Lie algebra} \cite{Dissertationes,CGM07}. Lie proved that every such a Lie algebra on the real line is locally diffeomorphic around a generic point to a
Lie subalgebra of $\langle \partial_x,x\partial _x,x^2\partial_x\rangle$ \cite{GKP92,1880}. This explains geometrically
why the general solution of a Riccati equation can be expressed in the form (\ref{sup}).

Several generalisations of the Riccati equations have appeared over the years \cite{KKW01,Kr05,PW07}. For instance, research has
been performed on the complex Riccati equations \cite{KR09}
\begin{equation}\label{RiccC}
\frac{dz}{dt}=a(t)+b(t)z+c(t)z^2,\qquad z\in \mathbb{C},
\end{equation}
with $a(t)$, $ b(t)$ and $c(t)$ being arbitrary complex $t$-dependent functions.
Some works have been devoted to studying their periodic solutions \cite{Ca97} and applications in physics
\cite{Ca13,CK13,CSCR15,KK08,Sc12}.

Riccati equations were lately generalised to the realm of quaternionic numbers $\mathbb{H}$ by defining \cite{PW07}:
\begin{equation}\label{RiccH}
\frac{dq}{dt}=a(t)+b(t)q+qc(t)+qd(t)q,\qquad q\in\mathbb{H}.
\end{equation}
In this case, $a(t)$, $b(t)$, $c(t),d(t)$ are arbitrary $t$-dependent quaternionic-valued functions. These
equations have been studied  mathematically in \cite{PW07}. Other quaternionic generalizations of Riccati equations have found applications in fluids physics and reducing multidimensional Schr\"odinger equations \cite{KKW01,Kr05}. As far as the authors know, we here provide  the first physical application of (\ref{RiccH}) by showing that these equations appear in the so-called quaternionic quantum mechanics \cite{Em63,FJS62,Ka60}. Particular cases of quaternionic Riccati equations, e.g. linear ones, appear in \cite{GR15}.

As Riccati equations over $\mathbb{R}$ and $\mathbb{C}$ have quadratic terms, it is reasonable that their generalization to quaternions should also content such a term. Since quaternions are not commutative, the quadratic term in (\ref{RiccH}) is not uniquely determined: we could also have
 $d(t)q^2$ or/and $q^2d(t)$ \cite{PW07}. A reason to use  $qd(t)q$ is
that (\ref{RiccH}) has an analytic extension to the one-point compactification, $\bar{\mathbb{H}}$, of quaternions  \cite{PW07}.

As a first main result, this work introduces a new generalization of Riccati equations to the space $\mathbb{O}$ of octonionic numbers: the
{\it octonionic Riccati equations}. Since the Hurwitz Theorem \cite{Ra93} states that every finite-dimensional normed division
algebra is isomorphic to $\mathbb{R}$, $\mathbb{C}$, $\mathbb{H}$ or $\mathbb{O}$ and we have a natural embedding of normed division algebras
$\mathbb{R}\subset\mathbb{C}\subset \mathbb{H}\subset\mathbb{O}$, then octonionic
Riccati equations retrieve Riccati equations over $\mathbb{R},\mathbb{C},\mathbb{H}$ as particular cases.

We prove that octonionic Riccati equations are Lie systems. The lack of commutativity of octonions leads to
several other possible definitions of octonionic Riccati equations, but we show that they do not lead to Lie systems. As a by-product we demonstrate that
the same occurs with the alternative definitions of quaternionic Riccati equations proposed in \cite{PW07} giving rise to a second reason to define quaternionic Riccati equations as  in (\ref{RiccH}).

We define Riccati equations over finite-dimensional normed division algebras (NDA Riccati equations) and we prove that they can be considered as sub-cases of conformal Riccati equations. This leads to show that  a Riccati equation over a finite-dimensional normed division algebra $A$ admits a Vessiot--Guldberg Lie algebra, $V_{{A}}$, of conformal vector fields with respect to a Euclidean metric and $V_{A}\simeq\mathfrak{so}(\dim A+1,1)$, where  $\mathfrak{so}(p,q)$ stands for the Lie algebra of the indefinite special orthogonal group $SO(p,q)$ (see \cite{HN12} for details). Importantly, Lie algebras $\mathfrak{so}(p,1)$ are simple Lie algebras \cite{Tam1999,VKK92} and  $V_{{A}}$ is the smallest common Vessiot--Guldberg Lie algebra for all Riccati equations over $A$. These results are resumed in Table \ref{tableVG}.

Since Vessiot--Guldberg Lie algebras of NDA Riccati equations become extremely large, standard
differential
and algebraic methods for studying Lie algebras, e.g. the Killing form, become highly inefficient. Instead,
 we use differential calculus over  normed division algebras and conformal geometry.

The fact that NDA Riccati equations are Lie systems allows for the use of many techniques to describe their general solutions, superposition rules, constants of motion, geometric invariants, Lie symmetries, etcetera \cite{AW,BHLS15,EHLS16}. It is worth noting that geometric techniques, e.g. symplectic, Poisson, $k$-symplectic or Jacobi structures, may be applied to study their properties \cite{BCHLS13,LTV15,LV15,GL12}.

It is well known that Riccati equations can be described as a projection of a linear system on $\mathbb{R}^2_\times:=\mathbb{R}^2\backslash\{(0,0)\}$ (see \cite{In44} for details). We extend this result to the octonionic Riccati equations and find that they  can also be extended to the octonionic projective line. A similar result applies for all NDA Riccati equations.

\begin{table}[h] {\footnotesize
 \noindent
\medskip
\noindent\hfill
 \begin{tabular}{ |c| c |  c|}
\hline
&  &\\[-1.9ex]
A&Vessiot--Guldberg Lie algebra & Dimension
\\[+1.0ex]
\hline
 &  &\\[-1.9ex]
$\mathbb{R}$&
$\mathfrak{conf}(\mathbb{R})\simeq \mathfrak{so}(2,1)$ & $3$
\\[+1.0ex]
$\mathbb{C}$&
$\mathfrak{conf}(\mathbb{R}^2)\simeq \mathfrak{so}(3,1)$ & $6$
\\[+1.0ex]
$\mathbb{H}$ &
$\mathfrak{conf}(\mathbb{R}^4)\simeq \mathfrak{so}(5,1)$&$15$
\\[+1.0ex]
$\mathbb{O}$&
$\mathfrak{conf}(\mathbb{R}^8)\simeq \mathfrak{so}(9,1)$ &$45$
\\[+1.0ex]
\hline
 \end{tabular}
\hfill}
\medskip
\caption{{\small 
{\footnotesize Vessiot--Guldberg Lie algebras for NDA Riccati equations. We write
$\mathfrak{conf}(\mathbb{R}^n)$ for the finite-dimensional Lie algebra of conformal vector fields of the Euclidean space $\mathbb{R}^n$.
}}}
\label{tableVG}
\end{table}

Subsequently, the study of the existence of Vessiot--Guldberg Lie algebras of Hamiltonian vector fields for particular types of octonionic and quaternionic Riccati equations is addressed.
The Lie algebras of associated Hamiltonian functions are given, which can be useful to obtain in a geometric way their superposition rules with the so-called co-algebra method and to study further their properties and constants of motion \cite{BCHLS13}.

The structure of the paper goes as follows. Section 2 is devoted to describing the fundamental properties of finite-dimensional normed
division algebras. Conformal  vector fields and M\"obius transformations are briefly presented in Section 3. Section 4 introduces the most important properties of Lie systems. We define octonionic Riccati equations and NDA Riccati equations in Section 5. In Section 6 it is proved that NDA Riccati equations are a particular type of conformal Riccati equations.
We show that other alternative definitions of octonionic and quaternionic Riccati equations  are not Lie systems in Section 7.  In Section 8 we prove that NDA Riccati equations can be recovered as the projection of a linear system. As a by-product, we show that NDA Riccati equations can be extended to projective lines over normed division algebras. In Section 9, we study the existence of symplectic structures turning types of NDA Riccati equations into Lie--Hamilton systems, i.e. Lie systems with a Vessiot--Guldberg Lie algebra of Hamiltonian vector fields relative to a Poisson structure \cite{CLS13}. Section 10 shows that quaternionic Riccati equations appear in the study of the so-called quaternionic Schr\"odinger equations. Finally, Section 11 summarises the work and gives some future lines of research.
\section{Normed division algebras}
Let us survey the theory of finite-dimensional normed division algebras to recall some results to be used hereafter. We also detail an intermediate result, Lemma 2.5, which allows us to use algebraic techniques to study differential geometric properties of NDA Riccati equations.

\begin{definition}\label{NorAlg}
(See \cite{Ba02,DM15,Ra93}) We call {\rm normed division algebra}  a finite-dimen\-sio\-nal $\mathbb{R}$-algebra $A$ equipped
with a unit, $1$, and a norm $\|\cdot \|:A\rightarrow [0,+\infty[$ such that
$
\|ab\|=\|a\| \|b\|,\forall a,b\in A.
$
\end{definition}
\begin{minipage}{5cm}
\noindent\includegraphics[scale=0.24]{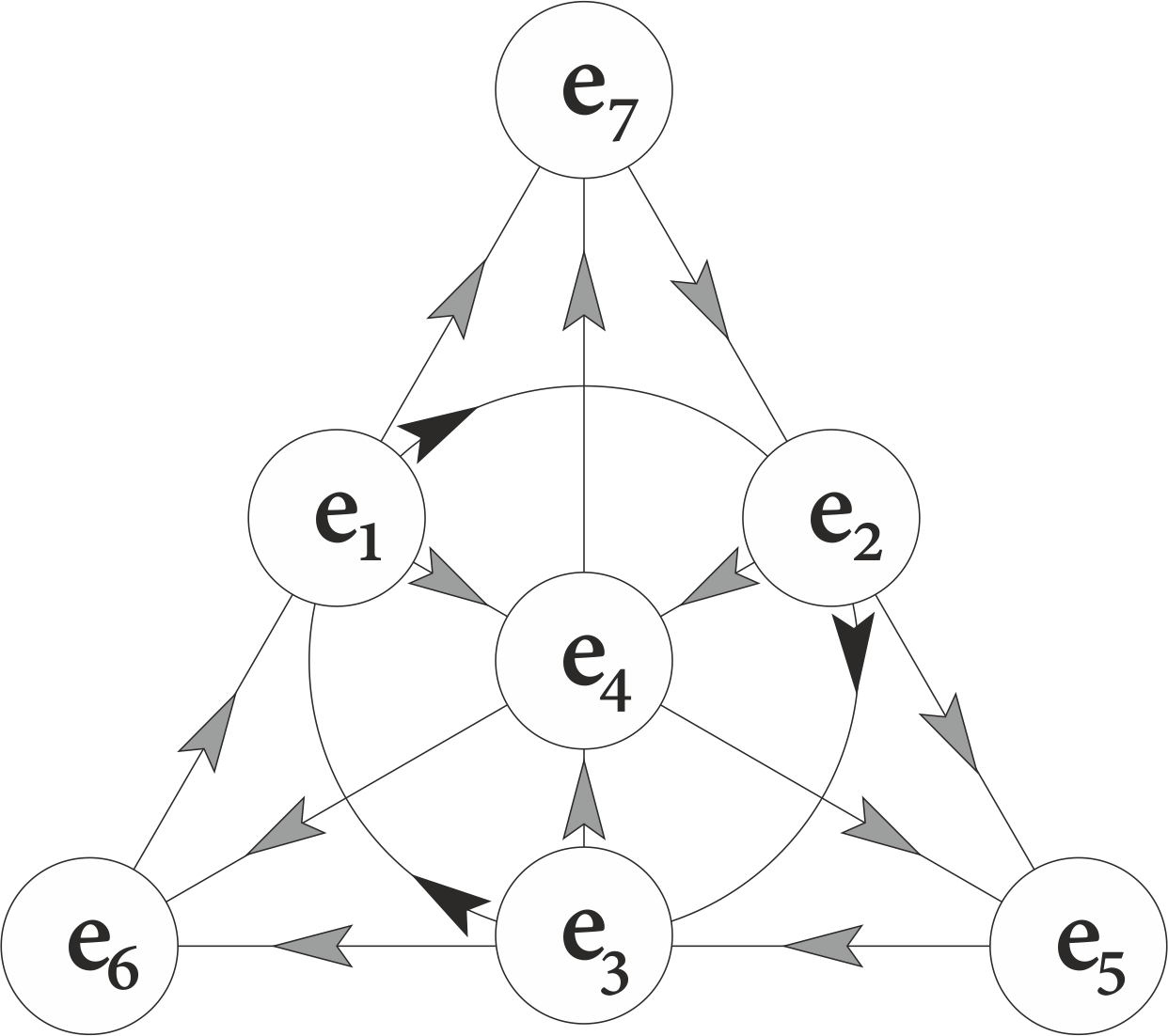}
\end{minipage}
\begin{minipage}{10cm}
\begin{example}
The space of octonions, $\mathbb{O}$, is an 8-dimensional normed division algebra. Its construction goes as follows. Given a basis $\{e_\alpha\}_{\alpha=0,\ldots,7}$ of $\mathbb{O}$, every $o\in \mathbb{O}$ takes the form $o=\sum_{\alpha=0}^7 o_\alpha e_\alpha$, with $o_\alpha\in\mathbb{R}$. The bilinear multiplication on $\mathbb{O}$ is determined by the multiplication on the elements
of our basis, which is given by the Fano diagram aside  (see \cite{Ba02} for details), and the relations $e_i^2:=-1$, for $i=1,\ldots,7$ and $e_0:=1$. We can endow $\mathbb{O}$ with a norm
$\|o\|:=[\sum_{\alpha=0}^7o_\alpha^2]^{1/2}$
turning it into a normed division algebra.
\end{example}
\end{minipage}

\begin{theorem}\label{Hurwitz}{(Hurwitz's theorem \cite{Ba02,DM15,Ra93})}
All normed division algebras are isomorphic to $\mathbb{R}$, $\mathbb{C}$, $\mathbb{H}$
or $\mathbb{O}$.
\end{theorem}

A normed division algebra $A$ admits an {\it inner product} $g:A\times A\rightarrow \mathbb{R}$ 
\begin{displaymath}
g(a,b):=\frac{\|a+b\|^2-\|a\|^2-\|b\|^2}{2},\qquad \forall a, b\in A.
\end{displaymath}
From its non-degeneracy, it follows that
\begin{equation}\label{Inner}
[\forall c\in A,\quad g(a,c)=g(b,c)] \quad \Longrightarrow \quad a=b.
\end{equation}
The space $A$ can be equipped with  the so-called {\it conjugation}
$*:A \ni a \mapsto a^*:=2g(a,1)-a\in A$, where $g(a,1)$ is understood as $g(a,1)1$. Observe that $aa^*=a^*a=\|a\|^2$. Then, $a\in A\backslash \{0\}$ has  inverse
$a^{-1}:={a^*}/{\|a\|^2}.$
\begin{example} In the case of the octonions, $e_0^*=e_0$ and $e_i^*=-e_i$ for $i=1,\ldots,7$. Thus,
every $o\in\mathbb{O}$ can be written in a unique way as $o=o_0+\bar{o}$, where $o_0:=[o^*+o]/2\in \mathbb{R}$ is called the {\rm scalar part} of $o$ and $\bar{o}:=[o-o^*]/2$ is referred to as {\rm its
vectorial part}.  Hence, $*:\mathbb{O}\ni o_0+\bar{o} \mapsto o_0-\bar{o}\in \mathbb{O}$.
\end{example}
Octonions form a particular type of composition algebra \cite{SC09,Ja58}. Therefore, they satisfy the following laws:
\begin{eqnarray}
g(ab, ac)=\|a\|^2g(b,c),\quad g(ac,bc)=g(a,b)\|c\|^2,\quad\quad \textrm{(scaling)},\quad\label{Scale}\\
g(ab, cd)=2g(a,c)g(b,d)-g(ad,cb),\quad\quad\quad\quad\textrm{(exchange)},\,\,\,\label{Exchange}\\
\,\,g(ab,c)=g(b,a^* c),\quad g(ab,c)=g(a,cb^*),\quad\quad\quad\quad \textrm{(braid)},\,\,\,\quad\label{braid}\\
a^*(a b)=\|a\|^2b=(b a)a^*,\quad a^{-1}(a b)=b=(b a)a^{-1},\quad\,\,\textrm{(inverse)},\quad\label{Inv}\\
a(a b)=a^2 b,\quad (b a)a = b a^2,\quad a(b a) = (a b) a,\quad\quad\,\,\textrm{(alternative)},\,\label{Alter}\\
(a b)(c a)\!=\![a(b c)]a,\, a[b(a c)]\!=\![(a b) a]c,\, [(b a)c]a\!=\!b[a (ca)],\textrm{(Moufang)}\label{Moufang}.
\end{eqnarray}

Although octonions are non-commutative and non-associative, they are alternative in view of (\ref{Alter}) and the
subalgebra generated by any two non-simultaneously zero octonions is isomorphic either to
$\mathbb{R}$, $\mathbb{C}$ or $\mathbb{H}$, which are associative. In view of this and Hurwitz's Theorem, we can frequently restrict our considerations to octonions - the most general normed division algebra.

Using (\ref{Exchange}), (\ref{braid})
and since $g(e_i,e_j)=0$
for $1\leq i\neq j\leq 7$, we get that
\begin{eqnarray*}
g(e_i(e_j o),\hat o)=-g(e_j o,e_i \hat o)=-2g(e_j,e_i)g(o,\hat o)+g(e_j \hat o, e_i o)=g(-e_j(e_i o),\hat o),\\
g((oe_i)e_j,\hat o)=-g(o e_i,\hat oe_j)=-2g(e_i,e_j)g(o,\hat o)+g(oe_j, \hat o e_i)=g(-(oe_j)e_i,\hat o),
\end{eqnarray*}
for all $o,\hat o\in \mathbb{O}$. Hence,
\begin{equation}\label{anti}
e_i(e_j o)=-e_j(e_io),\quad (oe_i)e_j=-(oe_j)e_i,\qquad 1\leq i\neq j\leq 7.
\end{equation}

Every normed division algebra forms a normed vector space over the reals.
Thus, $A$ becomes a complete metric space relative to ${\rm d}(a_1,a_2):=\|a_1-a_2\|$, where
$a_1$, $a_2\in A$. Since normed division algebras are isomorphic to a certain $\mathbb{R}^s$ as vector spaces,  they can then be understood as differentiable manifolds.

The next trivial lemma will be useful in following sections.
\begin{lemma}\label{lemma1}
 Given smooth curves $a_{(1)}(t), a_{(2)}(t), a_{(3)}(t)$ in $A$ with $\|a_{(3)}(t)\| $ $\neq 0$ for every $t\in\mathbb{R}$,
we have
$$
\begin{gathered}
\frac{d}{dt}[a_{(1)}(t)a_{(2)}(t)]=\frac{da_{(1)}(t)}{dt}a_{(2)}(t)+a_{(1)}(t)\frac{da_{(2)}(t)}{dt},\qquad
\\
\frac{da_{(3)}^{-1}(t)}{dt}=-a_{(3)}^{-1}(t)\frac{da_{(3)}(t)}{dt}a_{(3)}^{-1}(t).
\end{gathered}
$$
Similarly,
$\nabla_u (F_1F_2)=(\nabla_u F_1)F_2+F_1(\nabla_uF_2)$,
where $F_1,F_2:A\rightarrow A$ and ($\nabla_uG)(a)$ is the directional derivative in the direction
$u\in A$ of the function $G:A\rightarrow A$ at $a\in A$.
\end{lemma}

 We now give the general procedure of projective lines of normed division algebras. This structure will be used to relate NDA Riccati equations to mostly linear systems of differential equations (we refer to \cite{Ba02} for details).

Given $(a_1,a_2)\in A^2_\times:=A^2\backslash\{(0,0)\}$, we define the equivalence relation
\begin{equation}\label{rel1}
\begin{gathered}
(a_1,a_2)\sim(\tilde{a}_1,\tilde{a}_2),\,\,a_2\neq0 \Leftrightarrow\exists \lambda\in A\backslash\{0\}:(\tilde{a}_1,\tilde{a}_2)=((a_1a_2^{-1})\lambda,\lambda),\\
(a_1,a_2)\sim(\tilde{a}_1,\tilde{a}_2),\,\,a_1\neq 0\Leftrightarrow\exists \lambda\in A\backslash\{0\}:(\tilde{a}_1,\tilde{a}_2)=(\lambda,(a_2a_1^{-1})\lambda).
\end{gathered}
\end{equation}
The above equivalence relation is well defined. Indeed, for points $(a_1,a_2)$ with $a_1a_2\neq 0$ both possibilities in (\ref{rel1}) are equivalent due to the inverse laws.
The  quotient space, ${\rm AP}^1:=A^2_\times/\sim$, is the so-called {\it $A$-projective line} and $[a_1:a_2]$ stands for the equivalence class related to $(a_1,a_2)$. The ${\rm AP}^1$ is a differentiable
manifold with the differentiable structure induced by the atlas given through the charts $\psi_2: D_2\ni [a_1:1] \mapsto a_1\in A\simeq\mathbb{R}^s$ and $\psi_1: D_1\ni [1:a_2] \mapsto a_2\in A\simeq \mathbb{R}^s$, with $D_2:=\{[a_1:1]\in A{\rm P}^1:a_1\in A\}$ and $D_1:=\{[1:a_2]\in A{\rm P}^1:a_2\in A\}$. The
transition map between these charts, $\phi_A:\psi_1(D_1\cap D_2)\ni a\mapsto a^{-1}\in \psi_2(D_1\cap D_2)$,  is a diffeomorphism because $0\notin \psi_1(D_1\cap D_2)\cup \psi_2(D_1\cap D_2)$.

Let us also prove that the associated canonical projection map $\pi_A:A^2_\times\rightarrow A {\rm P}^1$ is smooth.  Consider the charts $\Psi_i:p\in U_i\subset A^2_\times\mapsto (a_1,a_2)\in A\times A\simeq \mathbb{R}^{2s}$, where $U_i:=\{(a_1,a_2)\in A^2:a_i\neq 0\}$ for $i=1,2$. Since $\pi_A(U_i)=D_i$, the local expressions for $\pi_A$  read
\begin{equation*}\label{map}
\psi_i\circ\pi_A\circ \Psi^{-1}_i:\Psi_i(U_i)\rightarrow \psi_i(D_i); \qquad \psi_i\circ\pi_A\circ \Psi^{-1}_i(a_1,a_2):=\begin{cases}a_2a_1^{-1},&i=1,\\a_1a_2^{-1},&i=2.\end{cases}\!\!
\end{equation*}
This mapping is well defined and smooth.

The previous construction can be simplified for associative  normed division algebras, e.g. quaternions. In this case, the projective line can be defined simply by considering ${A}^2_\times:={A}^2\backslash\{(0,0)\}$ and the equivalence relation
\begin{equation*}\label{rel}
(a_1,a_2)\sim (\tilde{a}_1,\tilde{a}_2)\iff\ \exists \lambda\in {A}\backslash\{0\}: (\tilde{a}_1,\tilde{a}_2)=(\lambda a_1,\lambda a_2).
\end{equation*}
\section{Conformal vector fields and transformations}
Let us state some facts on conformal geometry to be used hereafter.
\begin{definition} (See \cite{Sc08}) Let $(N,g)$ and $(N',g')$ be Riemannian manifolds and~${\rm dim}\,N={\rm dim}\,N'$. The mapping
$\varphi:U\subset N\rightarrow V\subset N'$ of maximal rank is a {\rm conformal transformation} if there exists a
function $f_\varphi:U\rightarrow \mathbb{R}_{>0}$,
the {\it conformal factor} of $\varphi$, satisfying
$
\varphi^{*}g'=f_\varphi^2g.
$
\end{definition}

\begin{definition}Let $(N,g)$ be a Riemannian manifold. A vector field $X$ on $N$ is {\rm conformal vector
field} relative to $g$ when
$
\mathcal{L}_X g=f_Xg
$
for a certain $f_X\in C^{\infty}(N)$ with $\mathcal{L}_X$ being the Lie derivative relative to $X$.
The function $f_X$ is called the {\rm potential} of $X$. A {\rm Killing vector field} is a conformal vector
field with zero potential.
\end{definition}

Conformal  vector fields on $(N,g)$ form a Lie algebra, its {\it conformal Lie algebra}, denoted by $\mathfrak{conf}(N,g)$ and conformal transformations
form a Lie group denoted by ${\rm Conf}(N,g)$.
In this work we are mainly interested in Euclidean spaces on $N=\mathbb{R}^n$, for $n=1$, $2$, $4$ and $8$. Hence, if a metric is given, we will assume it to be Euclidean.

The conformal Lie algebra in the plane is known to be infinite-dimensional \cite[p. 27]{Kl09}. In this case, we write $\mathfrak{conf}(\mathbb{R}^2)$ for the finite-dimensional Lie algebra of conformal vector fields associated to global (up to a unique singular point) orientation-preserving conformal transformations.
We have the Lie algebra isomorphisms (cf. \cite[p. 21]{Kl09})
$
\mathfrak{conf}(\mathbb{R}^n)\simeq \mathfrak{so}(n+1,1),
$
where $\mathfrak{so}(n+1,1)$ is the indefinite special orthogonal Lie algebra related to the metric on $\mathbb{R}^{n+2}$ with signature $(n+1,1)$. It is worth noting  that the $\mathfrak{conf}(\mathbb{R}^n)$ are maximal in the Lie algebra of polynomial vector fields on $\mathbb{R}^n$ and simple \cite{BLM18735}.

It is well known that global (up to a unique singular point) orientation-preserving  conformal transformations on the
complex plane are the referred to as {\it M\"obius transformations} \cite[p. 32]{Sc08}, namely
 $$
{\mathbb{C}}\ni z\mapsto \frac{\alpha z+\beta}{\gamma z+\delta}\in \mathbb{C},\qquad\qquad  \alpha\delta-\beta\gamma=1,\quad \alpha,\beta,\gamma,\delta\in \mathbb{C}.
 $$
M\"obius transformations form a group, called the {\it M\"obius
 group} ${\rm Mb}(\mathbb{C})$, which is isomorphic to $PSL(2,\mathbb{C}):=
 SL(2,\mathbb{C})/\mathbb{Z}_2$ \cite{PenRin86}.
 There is a link between conformal mappings,
 Lorentz transformations and spinor transformations for complex numbers \cite{PenRin86}. This was generalized by Manogue and
 Drey  to 
 normed division algebras \cite{ManDra98}. In particular, octonionic M\"obius transformations are
 $$
\mathbb{O}\ni o\mapsto (\alpha o+\beta)(\gamma o+\delta)^{-1}\in \mathbb{O},\qquad \alpha\delta-\beta\gamma=1,
 $$
 where $\alpha$, $\beta$, $\gamma$, $\delta$ lie in possibly different complex subspaces of $\mathbb{O}$
 \cite{ManDra98}. A similar result
 appears for quaternions (see \cite{GL12,PS09} and references therein).

 To summarise, we can understand normed division algebras as Riemannian manifolds with a Euclidean metric leading to the norm of the division algebra.
 The corresponding conformal vector fields span finite-dimensional Lie algebras $\mathfrak{conf}(\mathbb{R}^n)$
for $n=1$, $2$, $4$ and $8$. The flows of these vector fields are then conformal transformations acting  as M\"obius transformations.
\section{Fundamentals on Lie systems}
In this section the notion of Lie systems and its characterisation in terms of $t$-dependent vector fields and Lie algebras is detailed. All these concepts will be necessary to study geometrically NDA Riccati equations.

 \begin{definition}
A $t$-{\rm dependent vector field} on $N$ is a map
 $
 X: \mathbb{R}\times N\ni (t,x)\mapsto X_t(x):=X(t,x)\in TN ,
 $
 such that $\pi\circ X=p$ for $\pi: TN\rightarrow N$ and $p: \mathbb{R}\times N\ni (t,x)\mapsto x\in N$.
\end{definition}

So, a $t$-dependent vector field amounts to a family of vector fields $\{X_t\}_{t\in\mathbb{R}}$.
\begin{definition}
 An {\rm integral curve} of a $t$-dependent vector field $X$ on $N$ is an integral curve $\gamma : \mathbb{R}\rightarrow
 \mathbb{R}\times N$ of
 the vector field $\partial/\partial t + X(t, x)$ on $\mathbb{R}\times N$.
\end{definition}

Every $t$-dependent vector field admits a reparameterisation $\bar t=\bar t(t)$ such that
 $\gamma(\bar{t})=(\bar{t}, x(\bar{t}))$
 and the system describing its integral curves becomes
 \begin{equation}\label{tcurv}
 \frac{d(p\circ\gamma)}{d\bar{t}}(\bar{t})=(X\circ\gamma)(\bar{t}).
 \end{equation}

 Conversely, every system in normal form can be brought into the above form for a unique $t$-dependent vector field.
 This motivates to use $X$ to represent a $t$-dependent vector field and the system for its integral curves.

\begin{definition}
 Given a $t$-dependent vector field on $N$ of the form
$$
X(t,x)=\sum_{i=1}^nX_i(t,x)\frac{\partial}{\partial x_i},
$$
 its {\rm associated system} is the system determining its integral
curves, namely
$$
\frac{dx_i}{dt}=X_i(t,x),\qquad i=1,\ldots,n=\dim N.
$$
\end{definition}
\begin{example}A Riccati equation \!(\ref{RiccR})\! is related to\! the $t$-dependent vector field
\begin{equation}\label{RiccX}
X(t,x)=(b^{-}(t)+b^{(0)}(t)x+b^{+}(t)x^2)\frac{\partial}{\partial x},
\end{equation}
where we set  $b^{-}(t):=a(t)$, $b^{(0)}(t):=b(t)$ and $b^{+}(t):=c(t)$. This notation will be appropriate to relate NDA Riccati equations to
graded Lie algebras.
\end{example}

\begin{definition}
 The {\rm minimal Lie algebra} of a $t$-dependent vector field $X$ is the smallest real
Lie algebra of vector fields, $V^X$, containing $\{X_t\}_{t\in\mathbb{R}}$.
\end{definition}
\begin{example} The minimal Lie algebra of the $t$-dependent vector field (\ref{RiccX}) for $b^{-}(t)=b^{(0)}(t)=b^{+}(t)=a_0\in\mathbb{R}$ is $V^X=\langle a_0(1+x+x^2)\partial/\partial x\rangle$.
\end{example}

  \begin{definition}
  A {\rm superposition rule} depending on $m$ particular solutions for a system $X$ on $N$ is
a function
$
\Phi: N^m\times N\ni(x_{+} ,\ldots, x_{(m)},\lambda) \mapsto x\in N,
$
such that the
general solution of $X$ can be brought into the form $x(t)=\Phi(x_{+}(t),\ldots,x_{(m)}(t);\lambda)$,
where $x_{+}(t)$,$\ldots$, $x_{(m)}(t)$ is any generic family of particular solutions and $\lambda$ is a point
of $N$ to be related to initial conditions.
 \end{definition}
 \begin{example} In the case of Riccati equations, expression (\ref{sup}) shows that they admit a superposition rule
 $\Phi:\mathbb{R}^3\times\mathbb{R}\ni (x_{(1)},x_{(2)},x_{(3)},k)\mapsto x\in  \mathbb{R}$ of the form
 $$
x:=\frac{x_{(1)}(x_{(3)}-x_{(2)})+kx_{(2)}(x_{(3)}-x_{(1)})}{x_{(3)}-x_{(2)}+k(x_{(3)}-x_{(1)})}.
 $$
 Indeed, (\ref{sup}) implies that the general solution, $x(t)$, to any Riccati equation takes the form $x(t)=\Phi(x_{(1)}(t),x_{(2)}(t),x_{(3)}(t),k)$.
 \end{example}
\begin{definition}
 A nonautonomous system of first-order ordinary differential equations that admits a superposition rule is called a {\rm Lie system}.
\end{definition}

\begin{theorem}{\bf (Lie--Scheffers Theorem \cite{LS})}\label{Lie-Scheffers}
 A system $X$ admits a superposition rule if and only if $X =\sum_{\alpha=1}^rb_\alpha(t)X_\alpha$ for a certain family
 $b_1(t),\ldots, b_r(t)$ of
$t$-dependent real functions and vector fields $X_1,\ldots,X_r$ spanning an $r$-dimensional real Lie
algebra $V$.
\end{theorem}

The Lie algebra $V$ is called a {\it Vessiot--Guldberg Lie algebra} of $X$. The Lie--Scheffers Theorem can be
rewritten by saying that $X$ admits a superposition rule if and only if $V^X$ is finite-dimensional \cite{Dissertationes}.
\begin{example} We already know that Riccati equations are Lie systems, i.e. they have a superposition rule. Let us verify that they satisfy the conditions detailed by the Lie--Scheffers Theorem. The $t$-dependent vector field associated with Riccati equations takes the form
$$
X(t,x)=b^{-}(t)X^{-}+b^{(0)}(t)X^{(0)}+b^{+}(t)X^+,
$$
where $X^{-}:=\partial/\partial x$, $X^{(0)}:=x\partial/\partial x$ and $X^{+}:=x^2\partial/\partial x$ span a real Lie algebra $V_{\rm Ric}$ of vector fields isomorphic to $\mathfrak{sl}(2,\mathbb{R})$. Indeed,
$$
[X^{-},X^{(0)}]=X^{-},\qquad [X^{-},X^+]=2X^{(0)},\qquad [X^{(0)},X^+]=X^+.
$$
Note that $V_{\rm Ric}=\oplus_{k\in \mathbb{Z}}V^{(k)}$, for $V^{(-1)}:=\langle X^-\rangle,V^{(0)}:=\langle X^{(0)}\rangle,V^{(1)}:=\langle X^{+}\rangle$ and $V^{(k)}=0$ for $k\in \mathbb{Z}\backslash \{-1,0,1\}$. Hence, $[V^{(k_1)},V^{(k_2)}]\subset V^{(k_1+k_2)}$ for any $k_1,k_2\in \mathbb{Z}$, i.e. $V_{\rm Ric}$ is a $\mathbb{Z}$-graded Lie algebra. The notation of each $t$-dependent function details the gradding of the vector field to which is related.
\end{example}
\section{On the definition of octonionic Riccati equations}
Inspired by the form of the Riccati equations over $\mathbb{R}$, $\mathbb{C}$ and $\mathbb{H}$, we propose a Riccati equation over the octonions given by
\begin{equation}\label{OctRicc}
\frac{do}{dt}=b^{-}(t)+b^{0_L}(t)o+ob^{0_R}(t)+ob^{+}(t)o, \qquad o\in \mathbb{O},
\end{equation}
where $b^{-}(t),b^{0_L}(t),b^{0_R}(t),b^{+}(t)$ are arbitrary $\mathbb{O}$-valued $t$-dependent coefficients. The meaning of the notation employed for such functions will become fully clear posteriorly.
Since octonions are
alternative, the term $ob^{+}(t)o$ is well defined. Although this section is focused on (\ref{OctRicc}), the hereafter called {\it octonionic Riccati equations},
results can be easily extrapolated to any other {\it NDA Riccati equation}, namely a Riccati equation like (\ref{OctRicc}) where $o$ is assumed to belong to any normed division algebra.

The left-hand side of (\ref{OctRicc}) at $t=t_0$ can be understood as the limit $\lim_{t\rightarrow t_0}[o(t)-o(t_0)]/(t-t_0)$ in $\mathbb{O}$ relative to the topology induced by the norm of $\mathbb{O}$. Meanwhile, the right-hand side of (\ref{OctRicc}) is a mere element of $\mathbb{O}$ for each value of $t$.
Geometrically, the left-hand side of (\ref{OctRicc}) gives rise to an element of
$T_o\mathbb{O}$, i.e. the tangent vectors to a curve $\gamma:\mathbb{R}\ni t\mapsto o(t)\in \mathbb{O}$ passing through $o$ at $t$ and the right-hand side part can be understood for each value of $t$ as an element of $T_o\mathbb{O}$ by means of the mapping
$\lambda_o: \mathbb{O}\ni \tilde o\mapsto \lambda_o(\tilde o)\in T_o\mathbb{O}$ of the form
\begin{equation}\label{lift}
[\lambda_o(\tilde o)]f:=\frac{d}{d\tau}\bigg|_{\tau=0}f(o+\tau \tilde o),\qquad \forall f \in C^\infty(\mathbb{O}).
\end{equation}
This mapping is a linear isomorphism over $\mathbb{R}$ \cite{Is99}. Note also that (\ref{lift}) can immediately be extended to any NDA.

Above viewpoints allow us to investigate the properties of octonionic Riccati equations. One of their advantages is that they allow us to simplify calculations by using the algebraic structure of octonions. For instance, the following lemma will enable us to easily obtain Lie brackets of vector fields on $\mathbb{O}$. It can also be easily extended to any other normed division algebra.
\begin{lemma}\label{lem31} Given $F^{1)},F^{2)}: \mathbb{O}\rightarrow \mathbb{O}$, the vector fields on $\mathbb{O}$
defined by $Y_k(o):=\lambda_o[F^{k)}(o)]$, with $o\in\mathbb{O}$ and $k=1,2$, satisfy
\begin{equation}\label{grad}
[Y_1, Y_2](o)=\lambda_o\left\{[\nabla_{F^{1)}(o)}F^{2)}](o)-[\nabla_{F^{2)}](o)}F^{1)}](o)\right\},\qquad \forall o\in\mathbb{O}.
\end{equation}
\end{lemma}
\begin{proof}We assume that $F^{i)}(o)=(F_0^{i)}(o),\ldots,F_7^{i)}(o))$ for $i=1,2$. 
From the definitions of $\lambda_o$, $Y_1$ and $Y_2$, we get that $[Y_1,Y_2]f$ equals to
\[
\!\!\!\sum_{i,j=0}^7\!\left[F^{1)}_i\frac{\partial}{\partial o_i}\!\left(\!F^{2)}_j\frac{\partial f}{\partial o_j}\!\right)\!-\!F^{2)}_j
\frac{\partial}{\partial o_j}\left(\!F^{1)}_i\frac{\partial f}{\partial o_i}\right)\!\right]\!\! = \!\!\sum_{i,j=0}^7\!\left[
F^{1)}_i\frac{\partial F^{2)}_j}{\partial o_i}\frac{\partial f}{\partial o_j}\!-\!
F^{2)}_j\frac{\partial F^{1)}_i}{\partial o_j}\frac{\partial f}{\partial o_i}\right]
\]
for all $f\in C^{\infty}(\mathbb{O})$. Relabelling the summation indices, we obtain
\[
\begin{aligned}
([Y_1,Y_2]f)(o)&=\sum_{i,j=0}^7\left(F^{1)}_i\frac{\partial F^{2)}_j}{\partial o_i}\!-\!F^{2)}_i\frac{\partial F^{1)}_j}{\partial o_i}\right)
\!\frac{\partial f}{\partial o^j}(o)\\&=\lambda_o\{[\nabla_{F^{1)}(o)}F^{2)}-\nabla_{F^{2)}(o)}F^{1)}](o)\}f.
\end{aligned}\]
\end{proof}
\begin{lemma}
Octonionic Riccati equations are the associated systems with the $t$-dependent vector fields
\begin{equation}\label{newform}
\begin{gathered}
\!\!\!X\!\!=\!\!\sum_{i=0}^7(b^{-}_i(t)X_i^{-}\!+\!b^{+}_i(t)X_i^{+})\!+\!
b^{(0)}(t)X^{(0)}\!+\!\sum_{j=1}^7(b^{0_L}_j(t)X_j^{0_L}\!+\!b_j^{0_R}(t)X_j^{0_R}),
\end{gathered}
\end{equation}
with $b^{-}_i(t),b^{+}_i(t),b^{0_L}_j(t),b^{0_R}_j(t)$ being $t$-dependent real functions, $b^{(0)}(t)$  $ :=  b_0^{0_L}(t)+b_0^{0_R}(t)$, and
\begin{equation}\label{Field}
\begin{gathered}
X_i^{-}(o):=\lambda_o(e_i),\; X_i^{+}(o):=\lambda_o(o e_i o),\; X^{(0)}(o):=\lambda_o(o),\; \\
X_j^{0_L}(o):=\lambda_o(e_jo),\; X_j^{0_R}(o):=\lambda_o(oe_j),
\end{gathered}
\end{equation}
for $i=0,\ldots, 7$, $j=1,\ldots, 7$, the basis $\{e_0,\ldots,e_7\}$ of $\mathbb{O}$ and every $o\in \mathbb{O}$.
\end{lemma}

\begin{proof}
Writing the $\mathbb{O}$-valued functions of (\ref{OctRicc}) in coordinates, we obtain that (\ref{OctRicc}) amounts geometrically to
$$
\frac{do}{dt}=\lambda_o\left[\sum_{j=0}^7 (b^{-}_j(t)e_j+b^{0_L}_j(t)e_jo+b_j^{0_R}(t)oe_j+b_j^{+}(t)oe_jo)\right],
$$
where $b^{-}_j(t),b^{0_L}_j(t),$ $b_j^{0_R}(t),b_j^{+}(t)$ are $t$-dependent real functions and so they commute with all the elements of $\mathbb{O}$. Using the linearity of $\lambda_o$ over $\mathbb{R}$, we get
\begin{multline*}
\frac{do}{dt}=\sum_{j=0}^7[b^{-}_j(t)\lambda_o(e_j)+b^{+}_j(t)\lambda_o(oe_jo)]\\
+[b^{0_L}_0(t)+b^{0_R}_0(t)]\lambda_o(o)+
\sum_{j=1}^7[b^{0_L}_j(t)\lambda_o(e_jo)+b^{0_R}_j(t)\lambda_o(oe_j)].
\end{multline*}
Substituting (\ref{Field}) in the latter and using the definition of $b^{(0)}(t)$, we notice that the above system is the associated system of  (\ref{newform}).
\end{proof}
\section{NDA Riccati equations and Lie systems}

Let us prove that a NDA Riccati equation over a normed division algebra $A$ can be considered as a particular type of conformal Riccati equation with a Vessiot--Guldberg Lie algebra isomorphic to $\mathfrak{conf}(A)\simeq \mathfrak{so}(\dim A+1,1)$. 

Conformal Riccati equations are Lie systems on a metric space  $(\mathbb{R}^{p+q},$ $ \langle\cdot,
\cdot\rangle)$ of signature $(p,q)$ that posses a Vessiot--Guldberg Lie algebra of conformal vector fields isomorphic to $\mathfrak{so}(p+1,q+1)$
(see \cite{AW} for details). More specifically, conformal Riccati equations take the form
\begin{equation}\label{ConfRicc}
\frac{d\xi}{dt}=a(t)+\lambda(t) \xi+\Omega(t)\xi+c(t)\langle \xi,\xi\rangle -2\langle c(t),\xi\rangle
\xi,\qquad \xi \in \mathbb{R}^{p+q},
\end{equation}
where $a(t),c(t)\in \mathbb{R}^{p+q}$, $\lambda(t)\in\mathbb{R}$ and $\Omega(t)$ is a $(p+q)\times (p+q)$ real matrix
such that $\langle\Omega(t)\xi_1,\xi_2\rangle=-\langle \xi_1,\Omega(t)\xi_2\rangle$, for every $\xi_1,\xi_2\in\mathbb{R}^{p+q}$
and $t\in\mathbb{R}$ .

\begin{theorem} Every NDA Riccati equation
\begin{equation}\label{NDARicc}
\frac{da}{dt}=b^{-}(t)+b^{0_L}(t)a+ab^{0_R}(t)+ab^{+}(t)a,\qquad a \in
A,
\end{equation}
can be written in coordinates of a basis of $A$ in the form (\ref{ConfRicc}), with
$\lambda(t)=b^{0_L}_0(t)+b^{0_R}_0(t)$, the vectors $a(t), c(t)$ being the vector coordinate expressions of the elements $b^{-}(t), -[b^{+}(t)]^*$ respectively, and $\Omega(t)$ being the matrix related to the $\mathbb{R}$-linear operator $T_t:A\ni a\mapsto \overline{b^{0_L}(t)}a+a\overline{b^{0_R}(t)}\in A$.
\end{theorem}
\begin{proof} Let us take any elements $a,\hat a\in A$.  We have that
$$
\begin{aligned}
g(ab^{+}(t)a,\hat a)&=g(ab^{+}(t),\hat aa^*)=2g(a,\hat a)g(b^{+}(t),a^*)-
g( aa^*,\hat ab^{+}(t))\\
&=2g( a,\hat a)g( b^{+}(t),a^*)-g( a,a)g(1,\hat ab^{+}(t))\\
&=2g( a,\hat a)g([b^{+}(t)]^*,a)-g( a,a)g([b^{+}(t)]^*,\hat a)\\
&=g( 2g([b^{+}(t)]^*,a) a-[b^{+}(t)]^*g( a,a),\hat a),
\end{aligned}
$$
where we used relations (\ref{Scale}),(\ref{Exchange}),(\ref{braid}) and (\ref{Inv}). Since the above is true for
arbitrary $\hat a\in A$, we get from (\ref{Inner}) that
$$
ab^{+}(t)a=2g([b^{+}(t)]^*,a) a-[b^{+}(t)]^*g( a,a).
$$
Let us prove that $g(T_ta,\hat a)=-g(a,T_t\hat a)$ for every $t\in\mathbb{R}$. Indeed,
$$
\begin{gathered}
g(T_ta,\hat a)=g(\overline{b^{0_L}(t)}a+a\overline{b^{0_R}(t)},\hat a)=g(a,\overline{b^{0_L}(t)}^*\hat a+\hat a\overline{b^{0_R}(t)}^*)=-g(a,T_t\hat a).
\end{gathered}
$$
Using the above, equation (\ref{NDARicc}) can be rewritten as
$$
\frac{da}{dt}=b^{-}(t)+[b^{0_L}_0(t)+b^{0_R}_0(t)]a+T_ta
-[b^{+}(t)]^*g(a,a)+2g([b^{+}(t)]^*,a)a.
$$
Choose a  basis for $A$ and define $\xi$ to be the coordinate expression for $a$. Hence, $T_ta$ can be written in coordinates as $\Omega(t)\xi$, where $\Omega(t)$ is the matrix related to $T_t$ in the chosen basis. Fixing $\langle \cdot,\cdot\rangle$ in (\ref{ConfRicc}) to be $g(\cdot,\cdot)$, which implies $q=0$ and $p=\dim A$ because $g(\cdot,\cdot)$ is Euclidean, we obtain $0=\langle \Omega(t)\xi_1,\xi_2\rangle+\langle \xi_1,\Omega(t)\xi_2\rangle$ and (\ref{ConfRicc}) becomes the coordinate expression of (\ref{NDARicc}).

\end{proof}

\begin{corollary}
Every Riccati equation over a normed division algebra $A$ admits a Vessiot--Guldberg Lie algebra isomorphic to $\mathfrak{so}(\dim A+1,1)$.
\end{corollary}

\section{Other definitions of octonionic Riccati equations}

We now address the relevant question of the specific form of the octonionic Riccati equation (\ref{OctRicc}). More specifically, we establish that the addition of other quadratic terms, e.g. $e(t)o^2$ or $o^2f(t)$ for generic $t$-dependent functions $e(t)$ and $f(t)$, is not desirable since the $t$-dependent vector field associated with (\ref{OctRicc}) will no longer take values in a finite-dimensional real Lie algebra, i.e. the new defined octonionic Riccati equations will not be Lie systems.

Throughout this section  $\mathbb{O}$ is understood as an eight-dimensional manifold with a Euclidean metric given by
$g_\mathbb{O}=\sum_{i=0}^7do_i\otimes do_i.
$

\begin{lemma}\label{LemCon} The vector fields
\begin{equation}\label{VoVF}
X^{0_L}_1,\ldots, X^{0_L}_7,X_1^{0_R},\ldots,X^{0_R}_7,
\end{equation}
and their successive Lie brackets span a Lie algebra $V^{(0)}_\circ\simeq\mathfrak{so}(8)$ of Killing vector fields with respect to $g_\mathbb{O}$.
\end{lemma}
\begin{proof} First, let us prove that the norm on $\mathbb{O}$, i.e. $\|o\|=\sqrt{o^*o}$, is a first-integral of
the vector fields $X^{0_L}_i$ for $i=1,\ldots,7$.
We have
\begin{equation}\label{eqs1}
X^{0_L}_i\|o\|^2=2\|o\|X^{0_L}_i\|o\|,\qquad i=1,\ldots,7.
\end{equation}
Since $\|o\|^2=(o_0-\bar o)(o_0+\bar o)$, we also obtain
$$
X^{0_L}_i\|o\|^2=X^{0_L}_i\left[{(o_0-\bar o)(o_0+\bar o)}\right]=X^{0_L}_i[{o}^2_0- \bar o^2],\quad i=1,\ldots,7.
$$
Using that $X^{0_L}_io_0=(e_io )_0$, $\nabla_{X
^{0_L}_i}\bar o=\overline{e_i o }$ and the scaling law (\ref{Scale}), we get for $i=1,\ldots, 7$,
$$
X^{0_L}_i\|o\|^2=2o_0(e_i o)_0-\bar {o}\, \overline{e_i o}-\overline{e_i o}\bar o=2g(o,e_io)=2\|o\|^2g(e_0,e_i)=0.
$$
In view of (\ref{eqs1}), we get that $f(o)=\|o\|$ is a first-integral for the $X^{0_L}_i$ with $i=1,\ldots,7$ at every $o\neq 0$. Since $X^{0_L}(0)=0$, then $X^{0_L}f=0$ on the whole $ \mathbb{O}$\footnote{The function $f(o)=\|o\|$ is not differentiable at $o=0$, but it admits directional derivatives at this point in terms of the tangent vectors $X_i^{0_L}(0)$.}. The proof for the vector fields $X_i^{0_R}$
is
analogous. Since every vector field of $V^{(0)}_\circ$ is generated by the vector fields (\ref{VoVF}) and their
successive Lie brackets, the norm becomes a first-integral for all the elements of $V^{(0)}_\circ$.

In view of Table \ref{table1}, the vector fields (\ref{VoVF}) are linear in the chosen coordi-nates and therefore their successive
Lie brackets are also. The space of linear vector fields on $\mathbb{O}$ has finite-dimension, then $V^{(0)}_\circ$ must
be finite-dimensional. In turn, we can define a linear Lie group action $\Phi:G\times \mathbb{O}\rightarrow \mathbb{O}$
whose $G$ is a connected Lie group such that $T_eG\simeq V^{(0)}_\circ$. Since the fundamental vector fields leave the norm
invariant, we obtain that the norm is invariant under the action by the linear isomorphisms $\Phi_{g}:\mathbb{O}\ni o
\mapsto \Phi(g,o)\in \mathbb{O}$, for $g\in G$. Let us denote by $\Phi_{g*}$ the tangent map to $\Phi_{g}$.
Therefore,
\begin{multline}\label{metr}
g_{\mathbb{O}}(\Phi_{g*}v_1,\Phi_{g_*}v_2)=\frac 12[\|\Phi_{g*}(v_1+v_2)\|^2-\|\Phi_{g*}v_1\|^2-\|\Phi_{g*}v_2\|^2]\\
=\frac 12[\|(v_1+v_2)\|^2\!-\!\|v_1\|^2\!-\!\|v_2\|^2]=g_{\mathbb{O}}(v_1,v_2), \;  \forall v_1,v_2\in T_o\mathbb{O},\forall o\in\mathbb{O}.
\end{multline}
Hence, the elements of $\Phi_g$ are isometries of  $g_{\mathbb{O}}$ and they belong to $O(8)$. Since $G$ is connected
and the determinant of these mappings is a continuous function, it follows that $\Phi_e={\rm Id}$ implies that $\Phi_g\in
SO(8)$ for every $g\in G$ and the fundamental vector fields of $\Phi$ must be a Lie subalgebra of $\mathfrak{so}(8)$.

Let us prove that the fundamental vector fields of $\Phi$ generate the whole $\mathfrak{so}(8)$. Notice that for $i,j=1,\ldots,7$ we obtain
\begin{equation}\label{lleft}
\begin{array}{ll}
[X^{0_L}_i,X^{0_L}_j]&=\lambda_o[\nabla_{e_io}(e_jo)-\nabla_{e_jo}e_io]
=
\left\{\begin{array}{cc}-2\lambda_o[e_i(e_j o)],&i\neq j,\\0,&i=j,\end{array}\right.\end{array}
\end{equation}
where the first case on the right-hand expressions follows from (\ref{anti}). Thus, we generate a new family of 21
vector fields
$X^{(0)}_{ij}(o):=\lambda_o[e_i(e_jo)],$ $0<i<j\leq 7.
$
 \,\, From Table \ref{table1}, it follows that the vector fields
$\widetilde{X}_{0i}:= X^{0_L}_i+X^{0_R}_i$, with $i=1,\ldots, 7$,
and $\widetilde{X}_{ij}:=X_{ij}^{(0)}+X^{0_R}_{j\cdot i}$, for $1\leq i<j\leq 7$ and where we define $X^{0_R}_{j\cdot i}:=
\lambda_o[o(e_je_i)]$, are linearly independent and
generate a Lie algebra of vector fields of dimension 28 within $\mathfrak{so}(8)$. Since $\dim\mathfrak{so}(8)=8\cdot
7/2=28$, they must generate the whole Lie algebra. As $SO(8)$ is a compact Lie group, the exponential is surjective and
$G=SO(8)$. Since all vector fields of $V_\circ^{(0)}$ are fundamental vector fields of the action of $SO(8)$ and in view of (\ref{metr}), they are Killing vector fields with respect to $g_\mathbb{O}$.
\end{proof}

\begin{table}[h] {\footnotesize
 \noindent
\label{table1}
\caption{Linear vector fields appearing in octonionic Riccati equations}
\begin{center}
\begin{tabular}{| p{1cm} |  p{10.5cm}|}
\hline
 &\\[-1.9ex]
$X^{(0)}$&$o_0\frac{\partial}{\partial o_0}+o_1\frac{\partial}{\partial o_1}+o_2\frac{\partial}{\partial o_{2}}+o_3\frac{\partial}{\partial o_{3}}+o_4\frac{\partial}{\partial o_{4}}+o_5\frac{\partial}{\partial o_{5}}+o_{6}\frac{\partial}{\partial o_{6}}+o_7\frac{\partial}{\partial o_{7}}$ \\[+1.0ex]
$X_1^{0_R}$&$-o_1\frac{\partial}{\partial o_0}+o_0\frac{\partial}{\partial o_1}+o_3\frac{\partial}{\partial o_{2}}-o_2\frac{\partial}{\partial o_{3}}+o_5\frac{\partial}{\partial o_{4}}-o_4\frac{\partial}{\partial o_{5}}-o_7\frac{\partial}{\partial o_{6}}+o_{6}\frac{\partial}{\partial o_{7}}$ \\[+1.0ex]
$X_2^{0_R}$&$-o_2\frac{\partial}{\partial o_0}-o_3\frac{\partial}{\partial o_1}+o_0\frac{\partial}{\partial o_{2}}+o_1\frac{\partial}{\partial o_{3}}+o_{6}\frac{\partial}{\partial o_{4}}+o_7\frac{\partial}{\partial o_{5}}-o_4\frac{\partial}{\partial o_{6}}-o_5\frac{\partial}{\partial o_{7}}$ \\[+1.0ex]
$X_3^{0_R}$&$-o_3\frac{\partial}{\partial o_0}+o_2\frac{\partial}{\partial o_1}-o_1\frac{\partial}{\partial o_{2}}+o_0\frac{\partial}{\partial o_{3}}+o_7\frac{\partial}{\partial o_{4}}-o_{6}\frac{\partial}{\partial o_{5}}+o_5\frac{\partial}{\partial o_{6}}-o_4\frac{\partial}{\partial o_{7}}$ \\[+1.0ex]
$X_4^{0_R}$&$-o_4\frac{\partial}{\partial o_0}-o_5\frac{\partial}{\partial o_1}-o_{6}\frac{\partial}{\partial o_{2}}-o_7\frac{\partial}{\partial o_{3}}+o_0\frac{\partial}{\partial o_{4}}+o_1\frac{\partial}{\partial o_{5}}+o_2\frac{\partial}{\partial o_{6}}+o_3\frac{\partial}{\partial o_{7}}$ \\[+1.0ex]
$X_5^{0_R}$&$-o_5\frac{\partial}{\partial o_0}+o_4\frac{\partial}{\partial o_1}-o_7\frac{\partial}{\partial o_{2}}+o_{6}\frac{\partial}{\partial o_{3}}-o_1\frac{\partial}{\partial o_{4}}+o_0\frac{\partial}{\partial o_{5}}-o_3\frac{\partial}{\partial o_{6}}+o_2\frac{\partial}{\partial o_{7}}$ \\[+1.0ex]
$X_6^{0_R}$&$-o_{6}\frac{\partial}{\partial o_0}+o_7\frac{\partial}{\partial o_1}+o_4\frac{\partial}{\partial o_{2}}-o_5\frac{\partial}{\partial o_{3}}-o_2\frac{\partial}{\partial o_{4}}+o_3\frac{\partial}{\partial o_{5}}+o_0\frac{\partial}{\partial o_{6}}-o_1\frac{\partial}{\partial o_{7}}$ \\[+1.0ex]
$X_7^{0_R}$&$-o_7\frac{\partial}{\partial o_0}-o_{6}\frac{\partial}{\partial o_1}+o_5\frac{\partial}{\partial o_{2}}+o_4\frac{\partial}{\partial o_{3}}-o_3\frac{\partial}{\partial o_{4}}-o_2\frac{\partial}{\partial o_{5}}+o_1\frac{\partial}{\partial o_{6}}+o_0\frac{\partial}{\partial o_{7}}$ \\[+1.0ex]
$X_1^{0_L}$&$-o_1\frac{\partial}{\partial o_0}+o_0\frac{\partial}{\partial o_1}-o_3\frac{\partial}{\partial o_{2}}+o_2\frac{\partial}{\partial o_{3}}-o_5\frac{\partial}{\partial o_{4}}+o_4\frac{\partial}{\partial o_{5}}+o_7\frac{\partial}{\partial o_{6}}-o_{6}\frac{\partial}{\partial o_{7}}$ \\[+1.0ex]
$X_2^{0_L}$&$-o_2\frac{\partial}{\partial o_0}+o_3\frac{\partial}{\partial o_1}+o_0\frac{\partial}{\partial o_{2}}-o_1\frac{\partial}{\partial o_{3}}-o_{6}\frac{\partial}{\partial o_{4}}-o_7\frac{\partial}{\partial o_{5}}+o_4\frac{\partial}{\partial o_{6}}+o_5\frac{\partial}{\partial o_{7}}$ \\[+1.0ex]
$X_3^{0_L}$&$-o_3\frac{\partial}{\partial o_0}-o_2\frac{\partial}{\partial o_1}+o_1\frac{\partial}{\partial o_{2}}+o_0\frac{\partial}{\partial o_{3}}-o_7\frac{\partial}{\partial o_{4}}+o_{6}\frac{\partial}{\partial o_{5}}-o_5\frac{\partial}{\partial o_{6}}+o_4\frac{\partial}{\partial o_{7}}$ \\[+1.0ex]
$X_4^{0_L}$&$-o_4\frac{\partial}{\partial o_0}+o_5\frac{\partial}{\partial o_1}+o_{6}\frac{\partial}{\partial o_{2}}+o_7\frac{\partial}{\partial o_{3}}+o_0\frac{\partial}{\partial o_{4}}-o_1\frac{\partial}{\partial o_{5}}-o_2\frac{\partial}{\partial o_{6}}-o_3\frac{\partial}{\partial o_{7}}$ \\[+1.0ex]
$X_5^{0_L}$&$-o_5\frac{\partial}{\partial o_0}-o_4\frac{\partial}{\partial o_1}+o_7\frac{\partial}{\partial o_{2}}-o_{6}\frac{\partial}{\partial o_{3}}+o_1\frac{\partial}{\partial o_{4}}+o_0\frac{\partial}{\partial o_{5}}+o_3\frac{\partial}{\partial o_{6}}-o_2\frac{\partial}{\partial o_{7}}$ \\[+1.0ex]
$X_6^{0_L}$&$-o_{6}\frac{\partial}{\partial o_0}-o_7\frac{\partial}{\partial o_1}-o_4\frac{\partial}{\partial o_{2}}+o_5\frac{\partial}{\partial o_{3}}+o_2\frac{\partial}{\partial o_{4}}-o_3\frac{\partial}{\partial o_{5}}+o_0\frac{\partial}{\partial o_{6}}+o_1\frac{\partial}{\partial o_{7}}$ \\[+1.0ex]
$X_7^{0_L}$&$-o_7\frac{\partial}{\partial o_0}+o_{6}\frac{\partial}{\partial o_1}-o_5\frac{\partial}{\partial o_{2}}-o_4\frac{\partial}{\partial o_{3}}+o_3\frac{\partial}{\partial o_{4}}+o_2\frac{\partial}{\partial o_{5}}-o_1\frac{\partial}{\partial o_{6}}+o_0\frac{\partial}{\partial o_{7}}$ \\[+1.0ex]
\hline
\end{tabular}
\medskip

\begin{tabular}{|    c|}
\hline\\[-0.3cm]
Additional vector fields $(0\leq i<j\leq 7)$ \\[+1.0ex]
\hline\\[-1.5ex]
 $\widetilde{X}_{ij}:=[X_i^{-},X_j^{+}]=X^{(0)}_{ij}+X^{0_R}_{j\cdot i}=2 o_i\frac{\partial}{\partial o_j}-2o_j\frac{\partial}{\partial o_i}$\\[+1.0ex]
 $X^{(0)}_{ij}(o):=\lambda_o(e_i(e_jo)),\qquad i,j=0,\ldots,7.$\\[+1.0ex]
 \hline
\end{tabular}

\noindent
\end{center}
\hfill}
\end{table}

\begin{proposition}
A differential equation of the form
\begin{equation}\label{ExtRicc}
\frac{do}{dt}=b^{-}(t)+b^{0_L}(t)o+ob^{0_R}(t)+e(t)o^2,
\end{equation}
for arbitrary  $t$-dependent $\mathbb{O}$-valued functions $b^{-}(t), b^{0_L}(t),b^{0_R}(t), e(t)$, is not a Lie system.
\end{proposition}

\begin{proof}
 Consider the linear space $V$ of linear vector fields on the coordinates $o_0,\ldots, o_7$. Each $X\in V$ act on $E:=\langle o_0,\ldots,o_7\rangle$ as an $\mathbb{R}$-linear   operator with matrix $M_X$ in the basis $\{o_0,\ldots, o_7\}$. This gives rise to  a Lie algebra isomorphism $\Xi:V\ni X\mapsto  -M_X\in {\rm Mat}(\mathbb{R}^8)$ (cf. \cite{ACL10}). If $M_X$ is traceless, we will say that $X$ is traceless.
 In view of the coordinate expressions given in Table \ref{table1}, the  $\Xi(X_\alpha^{0_L})$, $\Xi(X_\alpha^{0_R})$ are antisymmetric matrices for $\alpha=1,\ldots,7$. As $\Xi$ is a Lie algebra isomorphism and the $X_\alpha^{0_L}, X_\alpha^{0_R}$ and their successive Lie brackets span $V_\circ^{(0)}\simeq \mathfrak{so}(8)$ in view of Lemma \ref{LemCon}, then the space $\Xi(V_\circ ^{(0)})$ consists of antisymmetric matrices spanning a subspace isomorphic to $\mathfrak{so}(8)$. This Lie algebra is a maximal Lie subalgebra of $\mathfrak{sl}(8,\mathbb{R})$. The addition of any other element of $\mathfrak{sl}(8,\mathbb{R})$ to $\mathfrak{so}(8)$ causes that the new subset can only be contained in the Lie algebra $\mathfrak{sl}(8,\mathbb{R})$ (cf. \cite{Ch93,ORW86,SW84}). Similarly, since $\Xi$ is a Lie algebra isomorphism, the addition of any new traceless vector field to $V_\circ^{(0)}$ causes that the resulting space can only be contained in the Lie algebra of traceless linear vector fields.

The Lie brackets of the vector fields $\{X^{-}_\alpha\}_{\alpha=0,\ldots,7}$ and $\langle X_i^{(1,L)}\rangle_{i=0,\ldots,7}$ generate linear traceless vector fields associated to not antisymmetric matrices. So, the vector fields associated to (\ref{ExtRicc}) generate all linear traceless vector fields.
Additionally, since $X^{(0)}$ is not a traceless vector field, the vector fields related to (\ref{ExtRicc}) generate all linear vector fields.
A short calculation shows that
$$
\left[o_\beta\frac{\partial}{\partial o_\beta},\left[o_\alpha\frac{\partial}{\partial o_\alpha},X_\alpha^{(1,L)}\right]\right]=(-1)^{0^\beta}2o_\beta^2\frac{\partial}{\partial o_\alpha},\qquad \alpha\neq \beta,\quad \alpha,\beta=0,\ldots,7,
$$
which allow us to generate the vector fields
$
X^{(i,j)}:=o^2_j\frac{\partial}{\partial o_i}$, $i\neq j$ from the vector fields in (\ref{ExtRicc}). The latter vector fields generate vector fields with coefficients which are polynomial with an arbitrary degree.  Indeed,
$$
\left[o^2_\beta\frac{\partial}{\partial o_\delta},\left[o_\delta\frac{\partial}{\partial o_\beta},o_\beta^k\frac{\partial}{\partial o_\alpha}\right]\right]=ko_\beta^{k+1}\frac{\partial}{\partial o_\alpha},\,\, \alpha\neq \beta, \beta\neq \delta, \alpha\neq \delta,\,\, \alpha,\beta=0,\ldots,7,
$$
for $k=2,3,\ldots$
\end{proof}

The addition of a term $o^2f(t)$ with generic $t$-dependent $\mathbb{O}$-valued  function $f(t)$ leads to a similar result by using identical methods. It is straightforward to particularise the above proposition for quaternionic Riccati equations.

Let us finally remark that octonionic Riccati equations do not cover all conformal Riccati equations related to Riemannian metrics on $\mathbb{R}^8$. Octonionic Riccati equations demand only 14 linear traceless vector fields to be described, namely $X^{0_L}_\alpha$ and $X^{0_R}_\alpha$ with $\alpha=1,\ldots,7$. This implies that its definition only contains 14 $t$-dependent coefficients. Meanwhile, $\Omega(t)$ appearing in conformal Riccati equations has 28 independent $t$-dependent entries. Thus, octonionic Riccati equations give rise to a subcase of conformal Riccati equations. Nevertheless, Lemma \ref{LemCon}  states that the traceless vector fields in the octonionic Riccati equation generate a Lie algebra isomorphic to $\mathfrak{so}(8)$, which has dimension 28. Hence, octonionic Riccati equations  and conformal Riccati equations over $\mathbb{R}^8$ share the same Vessiot--Guldberg Lie algebra for generic coefficients.
\section{Linear systems and extensions for NDA Riccati equations}
It is well known that Riccati equations can be extended to the real projective
line $\mathbb{R}{\rm P}^1\simeq S^1$. Moreover, particular solutions to Riccati equations can be obtained from particular solutions to associated linear systems \cite{In44}. The same results cannot be so easily
generalised to quaternions and octonions because they are non-associative and non-commutative, respectively. In this section, we will review the geometric approach to Riccati equations and we will extend such an approach to general normed division algebras.

Let us survey above mentioned known results on Riccati equations over the reals. Consider the system 
on $\mathbb{R}^2_\times:=\mathbb{R}^2\backslash\{(0,0)\}$ of the form
\begin{equation}\label{linric}
\left\{\begin{aligned}
\frac{dx_1}{dt}&=a_{11}(t)x_1+a_{12}(t)x_2,\\
\frac{dx_2}{dt}&=-a_{21}(t)x_1-a_{22}(t)x_2,
\end{aligned}\right.
\end{equation}
where $(x_1,x_2)\in\mathbb{R}^2_\times$ and $a_{11}(t), a_{12}(t), a_{21}(t), a_{22}(t)$ are arbitrary $t$-dependent real functions.

We can retrieve the Riccati equation from the above system in the following way. Define the surjective mapping $\pi_2:(x_1,x_2)\in U_2\subset \mathbb{R}^2\mapsto x_1/x_2\in \mathbb{R}$, with $U_{2}:=\{(x_1,x_2)\in\mathbb{R}^2:x_2\neq 0\}$. If $(x_1(t),x_2(t))$ is a particular solution to (\ref{linric}) within $U_{2}$, then $x(t):=\pi_2(x_1(t),x_2(t))$ satisfies
$$
\begin{aligned}
\frac{dx}{dt}&=\frac{dx_1}{dt}\frac{1}{x_2}-\frac{x_1}{x_2^2}\frac{dx_2}{dt}
\\&=(a_{11}(t)x_1+a_{12}(t)x_2)\frac{1}{x_2}+\frac{x_1}{x_2^2}(a_{21}(t)x_1+a_{22}(t)x_2)
\\
&=a_{12}(t)+(a_{11}(t)+a_{22}(t))x+a_{21}(t)x^2.
\end{aligned}
$$
So every particular solution to system (\ref{linric}) originates a particular solution to the above Riccati equation. Since $\pi_2$ is surjective, every solution to a Riccati equation can be obtained from a particular solution of (\ref{linric}) with $x_2(t)\neq 0$.

Let us describe the above results as a particular case of a more general geometric procedure describing the linearisation of quaternionic Riccati equations and their extension to a projective quaternionic line. Next, the procedure will be generalized to octonions in a nontrivial way.

The generalization of system (\ref{linric}) to the quaternions looks like
\begin{equation}\label{linricqua}
\left\{\begin{aligned}
\frac{dq_1}{dt}&=a_{11}(t)q_1+a_{12}(t)q_2,\\
\frac{dq_2}{dt}&=-a_{21}(t)q_1-a_{22}(t)q_2,
\end{aligned}\right.
\end{equation}
for $(q_1,q_2)\in \mathbb{H}^2_\times:=\mathbb{H}^2\backslash\{(0,0)\}$ and the $t$-dependent  functions $a_{11}(t),a_{12}(t),$ $a_{21}(t),a_{22}(t)$ taking values in $\mathbb{H}$.
This system is associated with the $t$-dependent vector field on $\mathbb{H}_\times^2$ of the form
$$
X_{\mathbb{H}}(t,q_1,q_2)=\lambda_{(q_1,q_2)}[a_{11}(t)q_1+a_{12}(t)q_2,-a_{21}(t)q_1-a_{22}(t)q_2].
$$
As in the real case, if $(q_1(t),q_2(t))$ is a particular solution for the quaternionic system (\ref{linricqua}) for $q_2(t)\neq 0$, then $q(t):=\pi_\mathbb{H}(q_1(t),q_2(t))$, with $\pi_{\mathbb{H}}:\mathbb{H}^2_\times\rightarrow \mathbb{H}{\rm P}^1$ satisfies the quaternionic Riccati equation (\ref{RiccH}). Indeed, using the coordinate $\psi_2:[q_1:q_2]\in D_2\mapsto w_2:= q_1q_2^{-1}\in \mathbb{R}^4$ we can write $q(t)=q_1(t)q^{-1}_2(t)$ and
\begin{equation}\label{toct}
\begin{aligned}
\frac{dq}{dt}&=\frac{dq_1}{dt}q_2^{-1}-q_1[q_2^{-1}\frac{dq_2}{dt}q_2^{-1}]\\&
=[a_{11}(t)q_1+a_{12}(t)q_2]q_2^{-1}+q_1[q_2^{-1}\{a_{21}(t)q_1+a_{22}(t)q_2\}q_2^{-1}]
\\&=a_{12}(t)+a_{11}(t)q+qa_{22}(t)+qa_{21}(t)q.
\end{aligned}
\end{equation}

We obtain that in the coordinates $\{q_1,q_2\}$ for $D_2$, understood as a subset of $\mathbb{R}^8$, and $w_2$ for $\mathbb{H}{\rm P}^1$, understood as an element of $\mathbb{R}^4$, the $t$-dependent vector field $\pi_{\mathbb{H}*}{X_\mathbb{H}}$ looks like
$$
\pi_{\mathbb{H}*}(X_{\mathbb{H}}(t,q_1,q_2))=\lambda_{w_2}[a_{12}(t)+a_{11}(t)w_2+w_2a_{22}(t)+w_2a_{21}(t)w_2].
$$
And using coordinates for $w_1=q^{-1}$ an analogue result is obtained:
\begin{equation}\label{toct}
\begin{aligned}
\frac{dw_1}{dt}&=-w_1a_{12}(t)w_1-w_1a_{11}(t)-a_{22}(t)w_1-a_{21}(t).
\end{aligned}
\end{equation}

Hence, $\pi_{\mathbb{H}*}X_\mathbb{H}$ is well defined on the whole $\mathbb{H}{\rm P}^1$. In one set of coordinates, it becomes the quaternionic Riccati equations and, in the other, a related one. Thus, quaternionic Riccati equations can be extended to $\mathbb{H}{\rm P}^1\simeq S^4$ \cite{BFLPP02}. A similar result follows for $\mathbb{R}$ and $\mathbb{C}$ by applying the same approach.

The `linearisation' of octonionic Riccati equations is much more subtle because of the lack of associativity.  Due to this, the system (\ref{linric}) is generalised to the system on $\mathbb{O}_\times^2:=\mathbb{O}^2\backslash\{(0,0)\}$ of the form
\begin{equation}\label{linoct}\begin{array}{ll}
\left\{\begin{aligned}
\frac{do_1}{dt}&=[a_{11}(t)(o_1o_2^{-1})]o_2+a_{12}(t)o_2,\\
\frac{do_2}{dt}&=-a_{21}(t)o_1-o_2\{o_1^{-1}[(o_1o_2^{-1})a_{22}(t)]\}o_2,
\end{aligned}\right. &o_1o_2\neq 0,\\ \noalign{\medskip}
\left\{\begin{aligned}
\frac{do_1}{dt}&=a_{11}(t)o_1+a_{12}(t)o_2,\\
\frac{do_2}{dt}&=-a_{21}(t)o_1-a_{22}(t)o_2,
\end{aligned}\right. &o_1o_2= 0,\end{array}
\end{equation}
where $(o_1,o_2)\in\mathbb{O}^2_\times$ and $a_{ij}(t)$, with $i,j=1,2$, are $t$-dependent $\mathbb{O}$-valued functions.

The morphism $\pi_{\mathbb{O}}:\mathbb{O}^2_\times\rightarrow \mathbb{O}{\rm P}^1$ is differentiable and for particular solutions $(o_1(t),o_2(t))$ with $o_2(t)o_1(t)\neq 0$  of the previous system we obtain that  using the coordinate chart $\psi_2$ the curve $\pi_{\mathbb{O}}(o_1(t),o_2(t))$ reads
\begin{equation}\label{toct}
\begin{aligned}
\frac{do}{dt}&=\frac{do_1}{dt}o_2^{-1}-o_1[o_2^{-1}\frac{do_2}{dt}o_2^{-1}]
\\&=
\{[a_{11}(t)(o_1o_2^{-1})]o_2+a_{12}(t)o_2\}o_2^{-1}+o_1[o_2^{-1}\{a_{21}(t)
o_1
\\&\qquad +o_2\{o_1^{-1}[(o_1o_2^{-1})a_{22}(t)]\}o_2\}o_2^{-1}]
\\&=a_{12}(t)+a_{11}(t)o+oa_{22}(t)+oa_{21}(t)o,
\end{aligned}
\end{equation}
where we used Lemma \ref{lemma1} and (\ref{Moufang}). 
For quaternions, system (\ref{linoct}) reduces to a linear system due to associativity. For octonions, the lack of associativity causes the terms on $a_{11}(t)$ and $a_{22}(t)$ to be no longer linear. Therefore, we will restrict the equation to the case of $a_{11}(t)$ and $a_{22}(t)$ being real $t$-dependent functions, when (\ref{toct}) can be proved to be linear. In fact, 
$$
\begin{gathered}
[a_{11}(t)(o_1o_2^{-1})]o_2=a_{11}(t)[(o_1o_2^{-1})o_2]=a_{11}(t)o_1\\
o_2\{o_1^{-1}[(o_1o_2^{-1})a_{22}(t)]\}o_2=a_{22}(t)[o_2\{o_1^{-1}(o_1o_2^{-1})\}o_2]
=a_{22}(t)o_2.
\end{gathered}
$$
 We can now relate the $t$-dependent vector field $X_{\mathbb{O}}$
on $\mathbb{O}^2_\times$ associated with (\ref{linoct}) to a $t$-dependent vector field on the octonionic projective line.
We obtain that in the coordinates $\{o_1,o_2\}$ for $D_2$, thought of as a subset of $\mathbb{R}^8\times\mathbb{R}^8$, and $w_2$ for $\mathbb{O}{\rm P}^1$ the $t$-dependent vector field $\pi_{\mathbb{O}*}X_\mathbb{O}$ looks like
$$
\pi_{\mathbb{O}*}(X_{\mathbb{O}}(t,o_1,o_2))=\lambda_{w_2}[a_{12}(t)+a_{11}(t)w_2+w_2a_{22}(t)+w_2a_{21}(t)w_2].
$$
And using coordinates for $w_1$ an analogue result is obtained. Hence, we obtain that $\pi_*X_\mathbb{O}$ is well defined on the whole $\mathbb{O}{\rm P}^1$. In one set of coordinates, it becomes the $t$-dependent vector field associated with octonionic Riccati equations and  a related one in the other. Therefore, the previous result allows us to extend octonionic Riccati equations with real $t$-dependent functions $a_{11}(t)$ and $a_{22}(t)$ to $\mathbb{O}{\rm P}^1\simeq S^8$.

The whole octonionic Riccati equations can be extended to $S^8$. It is enough to observe that the vector fields needed to describe (\ref{toct}) with real $a_{11}(t)$ and $a_{22}(t)$ generate under commutations all elements of $\mathfrak{conf}(\mathbb{R}^8)\simeq\mathfrak{so}(9,1)$. Let us prove this fact.

\begin{proposition} The Lie brackets among the vector fields contained in 
$
V^{-}:=\langle X^{-}_0,\ldots,X^{-}_7\rangle,$ $ 
V^{+}:=\langle X^{+}_0,\ldots,X^{+}_7\rangle,
$
span the whole Vessiot--Guldberg Lie algebra for octonionic Riccati equations.
\end{proposition}
\begin{proof}
Observe that
$$
[X^{-}_i,X^{+}_j]=\lambda_o\{\nabla_{e_i}o(e_jo)-\nabla_{oe_jo}e_i\}=
\lambda_o\{e_i(e_jo)+o(e_je_i)\},\quad i,j=0,\ldots,7.
$$
Hence, $\widetilde X_{ij}:=[X^{-}_i,X^{+}_j]=X^{0_R}_{j\cdot i}+X^{(0)}_{ij}$, where $X^{0_R}_{j\cdot i}(o):=\lambda_o[o(e_je_i)]$  and $X^{(0)}_{ij}(o):=\lambda_o\{e_i(e_jo)\}$ for $0\leq i<j\leq 7$. Additionally, $[X^{-}_i,X^{+}_i]=2(-1)^{0^i}X^{(0)}$ for $i=0,\ldots,7$. 
  We therefore obtain a family of 29 $\mathbb{R}$-linearly independent vector fields
 $\widetilde X_{ij}$, with $0\leq i<j\leq 7$ and $i=j=0$, which take the form given in Table \ref{table1}. All these vector fields are linear. Since the coefficients of the vector fields $V^{+}$ are quadratic on the variables $o_0,\ldots, o_7$ and the coefficients of the vector fields $V^{-}$ are non-zero constants, we obtain that all the previous vector fields are linearly independent and span the whole Vessiot--Guldberg Lie algebra $V^\mathbb{O}$.

\end{proof}

As the vector fields related to (\ref{toct}) are differentiable, can be extended to $\mathbb{O}{\rm P}^1$, and generate along with their Lie brackets $V^\mathbb{O}$, then all such vector fields do also. Since $\pi_{\mathbb{O}*}$ is a Lie algebra morphism from  vector fields on $\mathbb{O}_\times^2$ onto $\mathbb{O}{\rm P}^1$ and the vector fields of $V^{-}$ and $V^{+}$ are projections of well-defined linear vector fields on $\mathbb{O}_\times^2$, then the $t$-dependent vector field associated with (\ref{toct}) can be written as the projection of a $t$-dependent linear vector field on $\mathbb{O}^2_\times$ giving rise to a linear system of differential equations linearising octonionic Riccati equations.  So, we have proved the following result.
\begin{theorem} Every octonionic Riccati equation can be extended to $\mathbb{O}{\rm P}^1$ and lifted to a linear system on $\mathbb{O}^2_\times$.
\end{theorem}
\section{Symplectic structures and NDA Riccati equations}
A relevant question on the study of Lie systems is the existence of a symplectic form turning the vector fields of the
Vessiot--Guldberg Lie algebra into Hamiltonian vector fields. We hereafter analyse the existence of such a symplectic
structure for a general Vessiot--Guldberg Lie algebra on a manifold and, in particular, in the case of octonionic Riccati
equations.

%

\begin{proposition} No NDA Riccati equation admits a Vessiot--Guldberg Lie algebra of Hamiltonian vector fields for generic $t$-dependent coefficients.
\end{proposition}
\begin{proof} 
Let us prove our claim by reduction to absurd. Assuming the NDA Riccati equation to be  a Hamiltonian system for generic $t$-dependent coefficients implies that there exists a symplectic form $\omega$ on $A$ turning the Vessiot--Guldberg Lie algebra $V^A$ of the NDA Riccati equation into Hamiltonian vector fields. In particular, $\mathcal{L}_{X^{-}_i}\omega=\mathcal{L}_{X^{(0)}}\omega=0$ for $i=0,\ldots,7$. Since the vector fields $\{X^{-}_i\}_{i=0,\ldots,7}$ commute between themselves and they are by assumption Hamiltonian vector fields, we obtain
$
\mathcal{L}_{X^{-}_k}[\omega(X_i^{-},X_j^{-})]=0$  for $i,j,k=0,\ldots,7$. Since the $X_i^{-}$, with $i=0,\ldots,7$, span $TA$, then the functions $\omega(X_i^{-},X_j^{-})$ are constants. As $\mathcal{L}_{X^{(0)}}\omega=0$ and $[X^{(0)},X^{-}_j]=-X^{-}_j$ with $j=0,\ldots,7$, then
$$
0=\mathcal{L}_{X^{(0)}}(\omega(X_i^{-},X_j^{-}))=
-2\omega(X_i^{-},X_j^{-}),\qquad i,j=0,\ldots,7.
$$
Then, $\omega=0$, which is against our initial assumption. Hence, $V^A$ is not a Lie algebra of Hamiltonian vector fields relative to any symplectic structure.
\end{proof}
As general NDA Riccati equations do not posses Vessiot--Guldberg Lie algebras of Hamiltonian vector fields, there are many geometric methods, e.g. the co-algebra methods for superposition rules \cite{BCHLS13}, which cannot be used. Nevertheless, particular NDA Riccati equations may  admit a Vessiot--Guldberg Lie algebra of Hamiltonian vector fields.
The case of complex Riccati equations with $t$-dependent real coefficients was considered in \cite{EHLS16}.
Let us turn now to the octonionic Riccati equation
\begin{equation}\label{realOct}
\frac{do}{dt}=a(t)+b(t)o+c(t)o^2,
\end{equation}
where $a(t)$, $b(t)$, $c(t)$ are real $t$-dependent functions. A Vessiot--Guldberg Lie algebra of (\ref{realOct})
is spanned by
\begin{equation}\label{octField}
\begin{gathered}
\!\!X^{-}_{\mathbb{O}}\!=\!\frac{\partial}{\partial o^0},\,\,
X^{(0)}_{\mathbb{O}}\!=\!\sum_{i=0}^7 o_i\frac{\partial}{\partial o_i},\,\,
X^{+}_{\mathbb{O}}\!=\!\left(o_0^2-\sum_{i=1}^7 o_i^2\right)\frac{\partial}{\partial o_0}+2o_0\sum_{i=1}^7o_i\frac{\partial}{\partial o_i}.
\end{gathered}
\end{equation}
Observe that
\begin{equation*}\label{bracom}
[X^{-}_{\mathbb{O}},X^{(0)}_{\mathbb{O}}]=X^{-}_{\mathbb{O}},\quad
[X^{-}_{\mathbb{O}},X^{+}_{\mathbb{O}}]=2X^{(0)}_{\mathbb{O}},\quad
[X^{(0)}_{\mathbb{O}},X^{+}_{\mathbb{O}}]=X^{+}_{\mathbb{O}}.
\end{equation*}
Therefore, $V^{RO}:=\langle X^{-}_{\mathbb{O}}, X^{+}_{\mathbb{O}}, X^{(0)}_{\mathbb{O}}\rangle \simeq \mathfrak{sl}(2,\mathbb{R})$. We shall find a symplectic form on $\mathbb{O}$
turning the vector fields (\ref{octField}) into (locally) Hamiltonian with respect to it, i.e.
$\mathcal{L}_X\omega_{\mathbb{O}}=0$ for every $X\in V^{RO}$.
We introduce the new local coordinates on $\mathbb{O}$ given by  $\{o_0, \rho:=\sqrt{\sum_{i=1}^7o_i^2},
\alpha, \beta, \gamma, \psi, \theta, \phi\}$, where the latter six variables form a local coordinate
system on the sphere $S_\rho^6$ of radius $\rho$. In the new coordinates
\begin{equation*}
\begin{gathered}
X^{-}_{\mathbb{O}}=\frac{\partial}{\partial o_0},\,\,
X^{(0)}_{\mathbb{O}}=o_0\frac{\partial}{\partial o_0}+\rho\frac{\partial}{\partial\rho},\,\,
X^{+}_{\mathbb{O}}=(o_0^2-\rho^2)\frac{\partial}{\partial o_0}+2o_0\rho\frac{\partial}{\partial \rho}.
\end{gathered}
\end{equation*}
We now define the symplectic form
\begin{equation}\label{octsympl}
\omega_{\mathbb{O}}:=\frac{do_0\wedge d\rho}{\rho^2}+d\alpha\wedge d\beta+d\gamma\wedge d\psi+d\theta\wedge d\phi.
\end{equation}
Some Hamiltonian functions of (\ref{octField}) with respect to (\ref{octsympl}) are
\begin{equation}\label{hamfun}
f_{X^{-}_{\mathbb{O}}}=-\frac{1}{\rho},\quad
f_{X^{(0)}_{\mathbb{O}}}=-\frac{o_0}{\rho},\quad
f_{X^{+}_{\mathbb{O}}}=-\frac{(o_0)^2+\rho^2}{\rho},
\end{equation}
and
$$
\{f_{X^{-}_{\mathbb{O}}},f_{X^{(0)}_{\mathbb{O}}}\}=-f_{X^{-}_{\mathbb{O}}},\quad \{f_{X^{-}_{\mathbb{O}}},f_{X^{+}_{\mathbb{O}}}\}=-2f_{X^{(0)}_{\mathbb{O}}},\quad \{f_{X^{(0)}_{\mathbb{O}}},f_{X^{+}_{\mathbb{O}}}\}=-f_{X^{+}_{\mathbb{O}}}.
$$

Using the same procedure as before, we obtain a symplectic form for the quaternionic Riccati equation on points with $\theta\neq 0$ and $\rho\neq 0$:
\begin{equation}\label{quasympl}
\omega_{\mathbb{H}}=\frac{dq_0\wedge d\rho}{\rho^2}+\sin\theta d\theta\wedge d\phi,
\end{equation}
where $\rho$, $\theta$ and $\phi$ are the spherical coordinates on $\mathbb{R}^3$. 

\section{Applications in quaternionic quantum mechanics}
Some attempts to generalize quantum mechanics to the realm of the quaternions
have been performed \cite{Em63,FJS62,Ka60}. Much has been done to establish whether this is an appropriate approach \cite{Ad85,Ad86}. We here show that
Riccati equations over $\mathbb{H}$ can be employed in quaternionic quantum mechanics \cite{physics1995quaternionic}.

The {\it quaternionic Schr\"odinger equation} \cite{LDN02} is a generalization of the typical Schr\"odinger equation of the form
\begin{equation}\label{QSE}
\frac{\partial \Phi}{\partial t}=\left\{\frac {\rm i}\hbar \left[\frac{\hbar^2}{2m}\nabla^2-V\right]+\frac{{\rm j}}{\hbar}W\right\}
\Phi,
\end{equation}
with $\Phi:\mathbb{R}^{1+3}\ni(t,{\bf r}) \mapsto \Phi(t,{\bf r})\in \mathbb{H}$ being a quaternionic-valued function, $W:\mathbb{R}^{1+3}\rightarrow \mathbb{C}$ and $V:\mathbb{R}^{1+3}\rightarrow \mathbb{R}$ being potentials,
$\nabla^2$ being the Laplacian on $\mathbb{R}^3$ relative to the Euclidean metric, and $1,{\rm i},{\rm j},{\rm k}$ being the elements $e_0,e_1,e_2,e_3$ generating $\mathbb{H}$. As usual, $\hbar, m$ are the
Planck constant and the mass of the particle, respectively. It worth noting that (\ref{QSE}) is $\mathbb{H}$-linear relative to multiplication on the right.
 
Let us study the solutions of (\ref{QSE}) on $\mathbb{R}^{1+1}$ taking the form of stationary states, namely
$
\Phi(x,t)=\Psi(x)\exp\left[-{\rm i}Et/{\hbar}\right],\, E\in\mathbb{R}.
$
Hence,
\begin{equation}\label{tindqS}
\left\{{\rm i}\left[\frac{\hbar^2}{2m}\nabla^2-V\right]+{{\rm j}}W\right\}\Psi+\Psi  {\rm i} E=0.
\end{equation}
Let us show that the solution for $E=0$ can be solved by using 
Riccati equations over $\mathbb{H}$. Define a new variable
$
u:=\frac{{\rm d}\Psi}{{\rm d}x} \Psi^{-1}.
$
Using Lemma \ref{lemma1} and (\ref{tindqS}), we obtain that for $E=0$:
$$
\frac{{\rm d}u}{{\rm d}x}=\frac{{\rm d}^2 \Psi}{{\rm d}x^2} \Psi^{-1}
-\frac{{\rm d} \Psi}{{\rm d}x} \Psi^{-1}\frac{{\rm d} \Psi}{{\rm d}x} \Psi^{-1}
=-u^2+\frac{2m{\rm k}}{\hbar^2}W+\frac{2m}{\hbar^2}V.
$$
The latter can be considered as a quaternionic Riccati equation of the form
\begin{equation}\label{RiccQuan}
\frac{{\rm d}u}{{\rm d}x}=-u^2+b(x),\qquad b(x):=\frac{2m{\rm k}}{\hbar^2}W+\frac{2m}{\hbar^2}V.
\end{equation}
From the solution of this equation,
 the  solution to  (\ref{QSE}) with $E=0$ follows by solving  ${\rm d}\Psi/{\rm d}x=u(x)\Psi$.

Let us determine the minimal Lie algebra for (\ref{RiccQuan}) when $\{b(x)\}_{x\in\mathbb{R}}$ spans the whole $\mathbb{H}$. We
rewrite the equation (\ref{RiccQuan}) as
$$
\frac{{\rm d}u}{{\rm d}x}=-X_0^{+}(u)+\sum_{i=0}^3b_i(x)X^{-}_i(u),
$$
where $X_0^{+}(u)=\lambda_u(u^2)$, $X_i^{-}(u)=\lambda_u(e_i)$ for $i=0,\ldots,3$ and $\{e_i\}_{i=0}^3$ is the
standard basis of
the quaternions. We have to determine the smallest Lie algebra $V$ of vector fields containing $X_0^{+},X_0^{-}$.
Observe that
$$
[X_0^{-},X_0^{+}](u)=\lambda_u(\nabla_{e_0}u^2)=2\lambda_u(u)\Rightarrow X^{(0)}:=\lambda_u(u)\in V.
$$
Additionally,
$
[X_i^{-},X_0^{+}](u)=\lambda_u(e_iu+ue_i)$, for $i=1,2,3$.
Let us denote $X_{0i}^{(0)}:=\frac{1}{2}\lambda_u(e_iu+ue_i)$ for $i=0,\ldots,3$. These vector fields also belong to $V$.
Next, 
$$
[X_{0i}^{(0)},X^{+}_0](u)=\lambda_u(\nabla_{\frac{1}{2}(e_iu+ue_i)}u^2-\nabla_{u^2}(e_iu+ue_i)/2)=\lambda_u(ue_iu)=X^{+}_i\in V,
$$
for $j=1,2,3.$ Finally, 
$$
[X_j^{-},X_i^{+}](u)=\lambda_u(e_j(e_iu)+u(e_ie_j))=\widetilde{X}_{ji},
$$
for $1\leq i<j\leq 3$. Hence, $V$ contains a basis of fifteen linearly independent vector fields 
and $V\simeq \mathfrak{so}(5,1)$.

\section{Conclusions and Outlook}
We have proposed a new type of Riccati equations over the octonions and normed division algebras. We have proved that this is the only possible generalisation being a Lie system. We
have showed that NDA Riccati equations can be written as particular types of conformal Riccati equations. We provided a new procedure to link NDA Riccati equations to a type of linear systems over normed division algebras and we demonstrated that NDA Riccati equations can be extended to projective  lines. As a byproduct, other results for different types of Riccati equations in the literature have been retrieved. We have also studied several types of NDA equations admitting compatible symplectic structures. Finally, quaternionic Riccati equations have been applied in quaternionic quantum mechanics.

This work can be considered as a first study on conformal Lie systems, namely Lie systems with a Vessiot--Guldberg Lie algebra of conformal vector fields. In the future we expect to  study the general properties of these Lie systems by analysing, for instance, the use of conformal structures to obtain their superposition rules, constants of motions and Lie symmetries. We also plan to investigate Riccati equations over composition and Clifford algebras as well as to develop a theory of manifolds modeled on such structures. Finally, we aim to look for new potential applications of our methods.

\section{Acknowledgements}
 J. de Lucas acknowledges partial support by the National Science Centre (POLAND) under a grant HARMONIA DEC-2012/04/M/ST1/0533.
 M. Tobolski acknowledges partial support from projects MTM2011-15725-E (Ministerio de Ciencia e Innovaci\'on, Spain)
 and ESF Fizyka Plus POKL.04.01. 02-00-034/11. S. Vilari\~no acknowledges partial support from MTM2011-15725-E and MTM2011-2585 (Ministerio de Ciencia e Innovaci\'on) and E24/1 (Gobierno de
Arag\'on).

\end{document}